\documentclass[envcountsame]{llncs}

\usepackage{comment}
\usepackage{graphicx}
\usepackage{subfig}
\usepackage{mathtools}

\spnewtheorem{observation}[theorem]{Observation}{\bfseries}{\itshape}

\begin{document}

\title{Simultaneous Embeddings with Vertices Mapping to Pre-Specified Points}
\author{Taylor Gordon}
\institute{University of Waterloo}
\maketitle

\begin{abstract}

We discuss the problem of embedding graphs in the plane with restrictions on the
vertex mapping. In particular, we introduce a technique for drawing planar
graphs with a fixed vertex mapping that bounds the number of times edges bend.
An immediate consequence of this technique is that any planar graph can be drawn
with a fixed vertex mapping so that edges map to piecewise linear curves with at
most $3n + O(1)$ bends each. By considering uniformly random planar graphs, we
show that $2n + O(1)$ bends per edge is sufficient on average.

To further utilize our technique, we consider simultaneous embeddings of
$k$ uniformly random planar graphs with vertices mapping to a fixed, common
point set. We explain how to achieve such a drawing so that edges map to
piecewise linear curves with $O(n^{1-\frac{1}{k}})$ bends each, which holds with
overwhelming probability. This result improves upon the previously best known
result of $O(n)$ bends per edge for the case where $k \geq 2$. Moreover, we
give a lower bound on the number of bends that matches our upper bound,
proving our results are optimal.

\end{abstract}

\section{Introduction}

Of fundamental importance to graph drawing is the problem of drawing graphs in
the plane with restrictions on how vertices and edges are embedded. Indeed,
discussions on \emph{planar embeddings}, where vertices map to points and
edges map to continuous non-crossing curves, were commensurate with
the introduction of graph theory \cite{biggs}.

Bridges and Prussian cities aside, investigation into the properties of planar
embeddings has been motivated by applications such as \emph{information
visualization} and \emph{VLSI circuit design} (see \cite{aggarwal},
\cite{battista}, \cite{mead}). These applications provide metrics for which
certain embeddings become aesthetically or functionally preferable. For example,
a situation might prefer that edges be drawn as straight lines.

A classic result of F\'ary \cite{fary} showed that all \emph{planar graphs}
permit embeddings in the plane where each edge maps to a straight line segment
(a result independently proven by Wagner \cite{wagner} and Stein \cite{stein}).
If we further restrict the vertices to map to points on an $(n-2) \times (n-2)$
grid, then a planar embedding can still be achieved with edges mapping to
straight line segments \cite{schnyder}.

On the other hand, if the vertex mapping is completely fixed, a straight-line
embedding does not always exist. In fact, it was shown by Pach and Wenger
\cite{pach} that if we require edges to be drawn as \emph{polygonal curves}
(piecewise linear curves) then there does not always exist an embedding with
$o(n^2)$ total bends. Their results went further to show that this lower bound
holds \emph{asymptotically almost surely} for a uniformly random planar graph on
n vertices; that is, the lower bound holds with probability 1 as n tends to
infinity.

Kaufmann and Wiese \cite{kaufmann} considered the case where the range of the
vertex mapping is restricted to a fixed point set $P$ of size $n$. They showed
that any planar graph can be embedded so that each vertex maps to a unique point
in P and each edge maps to a polygonal curve with at most 2 bends. This result
is optimal in that there exists point sets (points on a line, for example) for
which not all planar graphs can be drawn with edges bending at most once.

The problem of drawing graphs to minimize bends has also been discussed in regards
to \emph{simultaneous embeddings}. A simultaneous embedding is a drawing in the
plane of $k$ graphs $G_1, G_2, \dots, G_k$, each over a common vertex set $V$, such
that no two edges of one graph cross. The concept of a simultaneous embedding
with this terminology was introduced in \cite{lubiw}. A related result of
particular interest was discussed in \cite{erten} by Erten and Kobourov. They
considered the special case of constructing a simultaneous embedding for when
$k = 2$. Their results showed that 2 bends per edge suffice to construct a
simultaneous embedding of two planar graphs.

One aim of our paper is to consolidate the above results on embedding graphs
with restrictions on the vertex mapping into a single drawing technique.
Lemma~\ref{lem:main} establishes such a technique that optimally minimizes the
number of bends (up to constant factors). Moreover, for the case where the
vertex mapping is completely fixed, we give a result matching the constant
factor of 3 on the number of bends per edge that was given in \cite{badent}. An
advantage of our technique, however, is that it lends itself well to
probabilistic analysis. Given a fixed vertex mapping, our technique gives at
most $2n$ bends per edge on expectation for a \emph{uniformly random planar
graph}, by which we mean a graph sampled uniformly at random from the set of all
planar graphs over the vertex set $V = \{1,2,\dots,n\}$.

Another aim of our paper is to generalize our results to simultaneous
embeddings. Our goal is to simultaneously embed planar graphs $G_1, G_2, \dots,
G_k$, each over a common vertex set $V$, so that each vertex uniquely maps to
one of $n = |V|$ pre-specified points. Using Lemma~\ref{lem:main}, we give a
construction for which each edge bends $O(n^{1 - \frac{1}{k}})$ times with
\emph{overwhelming probability} if we assume that the $k$ graphs are sampled
uniformly at random; that is, each edge bends $O(n^{1 - \frac{1}{k}})$ times
with probability at least $1 - n^{-c}$ for any \emph{fixed} constant $c$.

We go further to show that our result on simultaneous embeddings is optimal using
information theory. That is, we use an encoding argument to prove a lower bound
that matches our upper bound.

The drawing technique relies fundamentally on results related to book
embeddings, which we introduce in Section~\ref{sec:books}. We describe the
drawing technique in Section~\ref{sec:technique}. Section~\ref{sec:applications}
applies the drawing technique to the case of embedding a uniformly random planar
graph with a fixed vertex mapping. The application of the drawing technique to
simultaneous embeddings is described in Section~\ref{sec:se}. The proofs of the
lower bounds are in Section~\ref{sec:lb}. 

\section{Book Embeddings}
\label{sec:books}

A well-known result regarding \emph{book embeddings} is that all Hamiltonian planar
graphs have \emph{book thickness} 2 (see \cite{bernhart}). A trivial consequence of
this result is that a Hamiltonian planar graph can be embedded in the plane so
that all vertices lie on a common line and all edges lie strictly above or below
this line, except at their ends. Observe that in such an embedding, each edge
can be drawn as a polygonal curve with at most 1 bend (see Fig.~\ref{fig:book1}
for an example).

\begin{figure}
  \centering
  \subfloat[Hamiltonian supergraph]{\label{fig:book1}\includegraphics[width=0.4\textwidth]{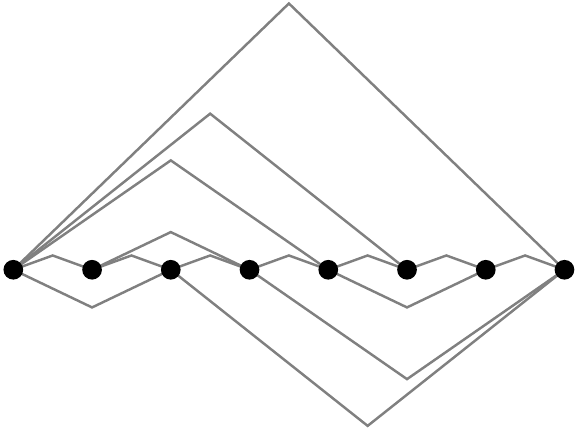}}
  \quad
  \subfloat[Original graph]{\label{fig:book2}\includegraphics[width=0.4\textwidth]{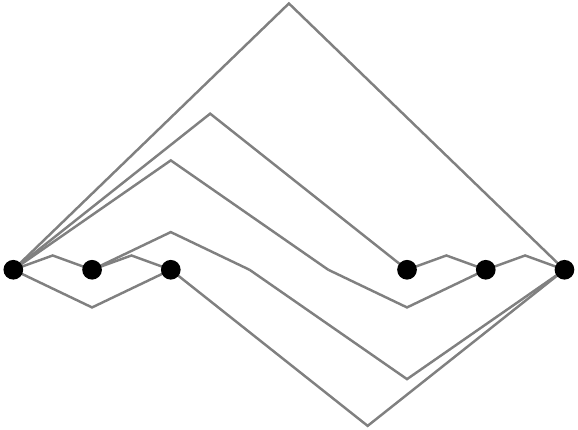}}
  \caption{The induced planar embedding of a graph from a book embedding.}
\end{figure}

Any planar graph can be augmented to become 4-connected by subdividing each edge
at most once and by adding additional edges. A classic theorem of
Tutte~\cite{tutte} showed that all 4-connected planar graphs are Hamiltonian.
It follows that we can always construct a Hamiltonian supergraph $G'$ of a
subdivision of a planar graph $G$ by subdividing each original edge at most
once\footnote{$G'$ can also be constructed in linear time by combining
results from \cite{biedl} and \cite{chiba}.}. From the Hamiltonian graph $G'$,
we can construct a book embedding (as in Fig.~\ref{fig:book1}), which induces an
embedding of the original graph $G$ (as in Fig.~\ref{fig:book2}).
Observation~\ref{obs:book} summarizes this embedding. Note that this embedding
and its construction has been frequently described in graph drawing literature
(as early as \cite{aggarwal}).

\begin{observation}
\label{obs:book}
A planar graph $G$ can be embedded in the plane so that
\begin{enumerate}
  \item all vertices lie on a common line,
  \item each edge bends at most once above the line, at most once below the
        line, and at most once on the line.
\end{enumerate}
\end{observation}

\section{Overview of the Drawing Technique}
\label{sec:technique}

Let $G = (V, E)$ be a planar graph, and suppose that $\gamma : V \to R^2$ is a
fixed vertex mapping. We  define $\delta$ to be any vector in $R^2$ such that
$\delta \cdot \gamma(u) = \delta \cdot \gamma(v)$, for $u,v \in V$, only if $u =
v$ (here $\cdot$ is the standard \emph{dot product} over $R^2$). That is, the
vertices in $V$ map under $\gamma$ to points at unique distances along the
direction of the vector $\delta$. Such a direction can trivially be seen to
always exist.

Suppose that $G$ is embedded as per Observation~\ref{obs:book}. For convenience,
we will refer to this embedding as the \emph{book embedding} of $G$ and the line
on which the vertices lie as the \emph{spine}. We can assume without loss of
generality that $\delta$ is aligned with the spine. Let $v_1, v_2, \dots, v_n$
be the vertices in $V$ as they occur along the direction of $\delta$ in the book
embedding. We relate the mapping $\gamma$ to this embedding of $G$ using
order-theoretic concepts.

\begin{definition}
\label{def:po}
Let $\prec$ be a partial order over $V$ such that $v_a \prec v_b$ if and only
if $a \leq b$ and $\delta \cdot \gamma(v_a) \leq \delta \cdot \gamma(v_b)$.
Similarly, let $\succ$ be a partial order over $V$ such that $v_a \succ v_b$ if
and only if $a \leq b$ and
$\delta \cdot \gamma(v_a) \geq \delta \cdot \gamma(v_b)$.
\end{definition}

Thus, a \emph{chain} with respect to $\prec$ is a set of vertices that occur
along $\delta$ in the same order in both the book embedding of $G$ and under the
mapping $\gamma$. On the other hand, the vertices in a chain with respect to
$\succ$ occur in the reversed order in the book embedding of $G$ from their
order under $\gamma$. Using this notation, we can state the effect of our
drawing technique as follows.

\begin{lemma}
\label{lem:main}
Suppose that $V$ is partitioned into $V_1, V_2, \dots, V_r$ so that $v_a \in
V_i$ and $v_b \in V_j$ satisfy
$\delta \cdot \gamma(v_a) < \delta \cdot \gamma(v_b)$ if $i < j$. Then, if $V_i$
forms a chain with respect to $\prec$ when $i$ is odd and a chain with respect
to $\succ$ when $i$ is even, we can embed $G$ in the plane with the vertex
mapping $\gamma$ using at most $3r + O(1)$ bends per edge. 
\end{lemma}

\begin{proof}
Without loss of generality, we can assume $\delta$ is directed horizontally.
Thus, we can assume that
\begin{enumerate}
  \item $v_1, v_2, \dots, v_n$ are the vertices in $G$ in the order they are
  mapped from left to right in the book embedding,
  \item $V_1, V_2, \dots, V_r$ map under $\gamma$ to the point sets
  $P_1, P_2. \dots, P_r$, respectively, such that all points in $P_i$ occur left
  of all points in $P_{i+1}$, for $i = 1,2,\dots,r-1$,
  \item for odd $i$, the vertices in $V_i$ map to points in $P_i$ with the same
  relative left-to-right order as they occur along the spine of the book
  embedding,
  \item for even $i$, the vertices in $V_i$ map to points in $P_i$ with the
  reverse relative left-to-right order as they occur along the spine of the book
  embedding.
\end{enumerate}
Thus, we can think of the vertex sets $V_1, V_2, \dots, V_r$ as mapping to
disjoint intervals $\Delta_1, \Delta_2, \dots, \Delta_r$ along the x-axis, each
(strictly) containing the points $P_1, P_2, \dots, P_r$ respectively. See
Fig.~\ref{fig:lem1} for an example of such a configuration. We will return to
this idea to show how to partially embed the edges in $G$ inside each interval,
but before doing so, we first introduce some terminology.

\begin{figure}
  \centering
  \includegraphics[width=1.0\textwidth]{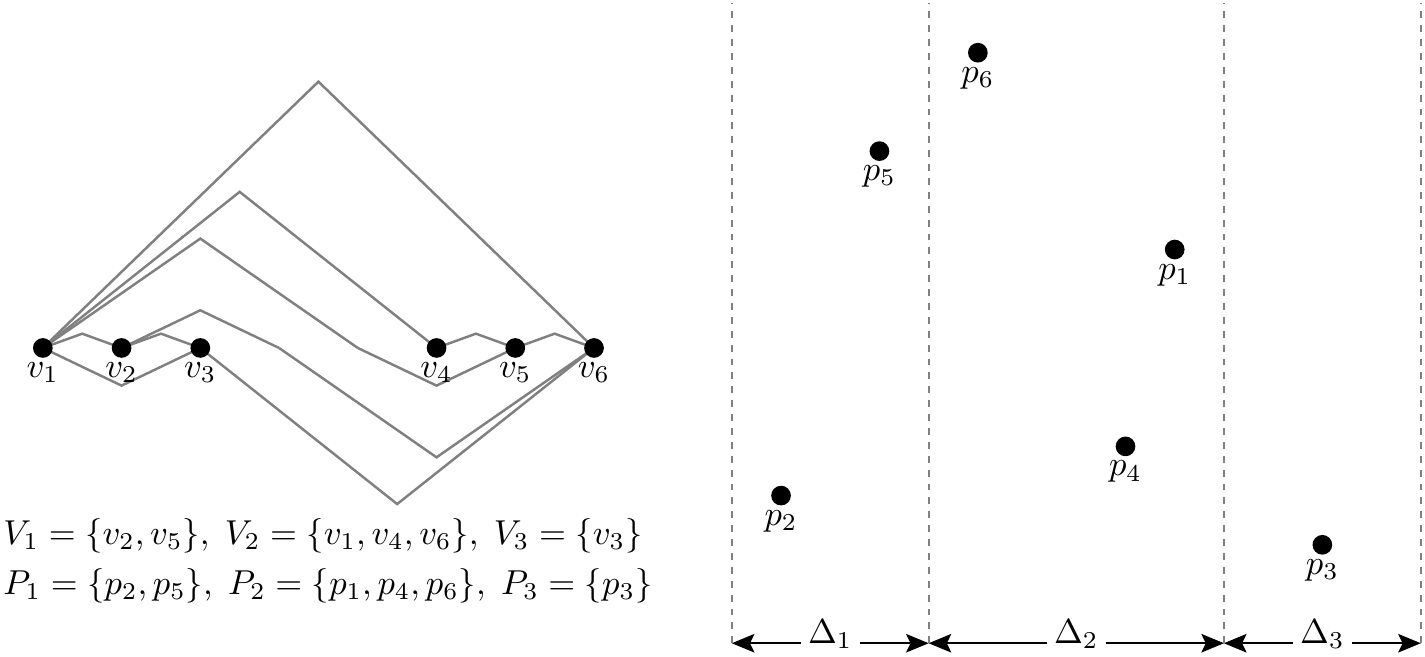}
  \caption{An example configuration for a graph on 6 vertices.}
  \label{fig:lem1}
\end{figure}

All points in the book embedding of $G$ that intersect with the spine either
correspond to a vertex in $G$ or a point at which an edge crosses the spine. Let
$\Gamma_1, \Gamma_2, \dots, \Gamma_h$ correspond to these points in the order
they occur along the spine from left to right. If $\Gamma_i$ corresponds to a
vertex $v$, then we define $\mathbf{vertex}(\Gamma_i) = v$. Furthermore, we
define $\mathbf{top}(\Gamma_i)$ to be the set of edges incident to $v$ that were
embedded on the top page and $\mathbf{bottom}(\Gamma_i)$ to be those embedded on
the bottom page. If $\Gamma_i$ corresponded to a point at which edges crossed
the spine, then $\Gamma_i$ unambiguously refers to this edge.

We now describe how to draw the edges in $G$ inside each of the intervals
$\Delta_1, \Delta_2, \dots. \Delta_r$. We first consider an interval $\Delta_i$,
for which $i$ is odd, with corresponding vertex set $V_i$ and point set $P_i$.
For $t = 1,2,\dots,h$, we will draw a set of vertical lines corresponding to
$\Gamma_t$ as follows (an example of the construction is shown in
Fig.~\ref{fig:lem3}).
\begin{enumerate}
  \item If $\mathbf{vertex}(\Gamma_t) = v \in V_i$, then we draw a vertical line
  above the point $\gamma(v)$ for each edge in $\mathbf{top}(\Gamma_t)$ and a
  vertical line below the point $\gamma(v)$ for each edge in
  $\mathbf{bottom}(\Gamma_t)$.  We then join the ends of these vertical lines to
  $\gamma(v)$ as is shown in Fig.~\ref{fig:lem2}.
  \item If $\mathbf{vertex}(\Gamma_t) = v \in V_j$, where $j > i$, then we draw
  a vertical line somewhere in the interval $\Delta_i$ for each edge in
  $\mathbf{top}(\Gamma_t)$.
  \item If $\mathbf{vertex}(\Gamma_t) = v \in V_j$, where $j < i$, then we draw
  a vertical line somewhere in the interval $\Delta_i$ for each edge in
  $\mathbf{bottom}(\Gamma_t)$.
  \item If $\Gamma_t$ did not correspond to a vertex, we draw a vertical line
  for its unique edge.
  \item All lines drawn for $\Gamma_t$ occur left of all lines drawn for
  $\Gamma_{t+1}$.
\end{enumerate}
\begin{figure}
  \centering
  \includegraphics[width=0.3\textwidth]{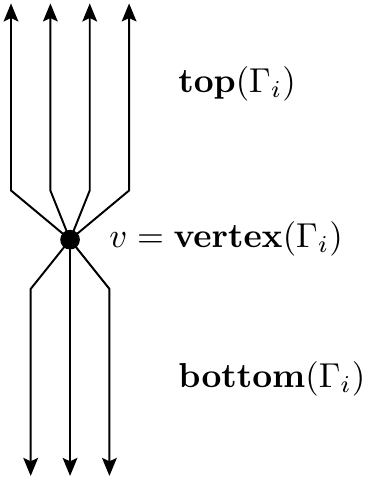}
  \caption{A demonstration of how we join the edges in $\mathbf{top}(\Gamma_i)$ and
           $\mathbf{bottom}(\Gamma_i)$ to the vertex $\mathbf{vertex}(\Gamma_i)$.}
  \label{fig:lem2}
\end{figure}
Note that the final condition enforces that the vertical lines are separated
into intervals, the first corresponding to edges from $\Gamma_1$, the second to
edges from $\Gamma_2$, and so forth. It is clear that the first four conditions
can easily be achieved. The last condition follows by the restriction on the
left-to-right order of the points in $P_i$. Indeed, the vertices in $V_i$ are
mapped to points in $P_i$ that occur from left to right in the same order as the
points on the spine of the book embedding that defined $\Gamma_1, \Gamma_2,
\dots, \Gamma_h$.

For an interval $\Delta_i$, for which $i$ is even, the procedure is symmetric.
The only difference is that the last condition is reversed so that the lines are
drawn corresponding to those from $\Gamma_h$ first, then those from
$\Gamma_{h-1}$, and so forth. It follows by the symmetric definition of $V_i$, for
even $i$, why this construction can be achieved.

Thus, we can repeat the above process for each \emph{$\Delta$-interval}
$\Delta_1, \Delta_2, \dots, \Delta_r$. After this procedure, each
$\Delta$-interval will have a set of lines extending upwards and a set of lines
extending downwards (some of which may correspond to the same edge). When $i$ is
odd, we say that the lines extending upward in $\Delta_i$ are \emph{entering}
$\Delta_i$ and those extending downward are \emph{leaving}. When $i$ is even,
the definitions are reversed. We make a few observations about the configuration
of these lines.
\begin{enumerate}
  \item When $i$ is odd, the lines in $\Delta_i$ define subintervals along the
  $x$-axis that correspond to $\Gamma_1, \Gamma_2, \dots, \Gamma_h$ from left to
  right.
  \item When $i$ is even, the lines in $\Delta_i$ define subintervals along the
  $x$-axis that correspond to $\Gamma_h,\Gamma_{h-1},\dots,\Gamma_1$ from left to
  right.
  \item For contiguous intervals $\Delta_i$ and $\Delta_{i+1}$, the lines
  leaving $\Delta_i$ for a particular $\Gamma_t$ correspond to the same set of
  edges as the lines entering $\Delta_{i+1}$ for $\Gamma_t$.
\end{enumerate}
The last observation follows by construction. Clearly it holds for any
$\Gamma_t$ that corresponded to an edge crossing the spine. Suppose instead that
$v = \mathbf{vertex}(\Gamma_t)$. If the lines leaving $\Delta_i$ corresponded to
the edges in $\mathbf{top}(\Gamma_t)$, then $v \notin V_j$ for all $j \leq i$,
implying that the lines entering $\Delta_{i+1}$ for $\Gamma_t$ also correspond
to $\mathbf{top}(\Gamma_t)$. On the other hand, if the lines leaving $\Delta_i$
corresponded to the edges in $\mathbf{bottom}(\Gamma_t)$, then $v \in V_j$, for
some $j \leq i$, implying that the lines entering $\Delta_{i+1}$ for $\Gamma_t$
also correspond to $\mathbf{bottom}(\Gamma_t)$.

\begin{figure}
  \centering
  \includegraphics[width=0.35\textwidth]{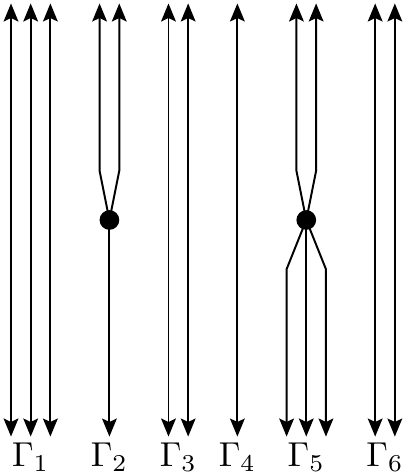}
  \caption{An example of the vertical lines drawn for an interval $\Delta_i$.}
  \label{fig:lem3}
\end{figure}

We proceed to show how to join the vertical lines from contiguous
$\Delta$-intervals. Let $B$ be an axis-aligned box containing all points in $P$
(where $P$ is the image of $V$ under $\gamma$). For even $i$, suppose we were to
rotate all vertical lines in the interval $\Delta_i$ clockwise by a small angle
$\epsilon$  so that the lines remain parallel and only leave the interval
$\Delta_i$ outside of the box $B$. Eventually, the lines drawn for each
$\Gamma_t$ in the interval $\Delta_i$ would intersect with the vertical lines
drawn for $\Gamma_t$ in the interval $\Delta_{i-1}$. We can then terminate the
lines drawn for $\Delta_{i-1}$ and $\Delta_i$ at these intersection points,
hence joining the lines drawn for $\Gamma_t$ in $\Delta_{i-1}$ and $\Delta_i$.
Similarly, if we consider the intersection between the lines extending from
$\Delta_i$ with the vertical lines extending upward from $\Delta_{i+1}$
(assuming $\Delta_{i+1}$ exists), we can again terminate these lines at the
points the lines from a common $\Gamma_t$ intersect. See Fig.~\ref{fig:lem4} for
an example of this procedure.

\begin{figure}
  \centering
  \includegraphics[width=0.7\textwidth]{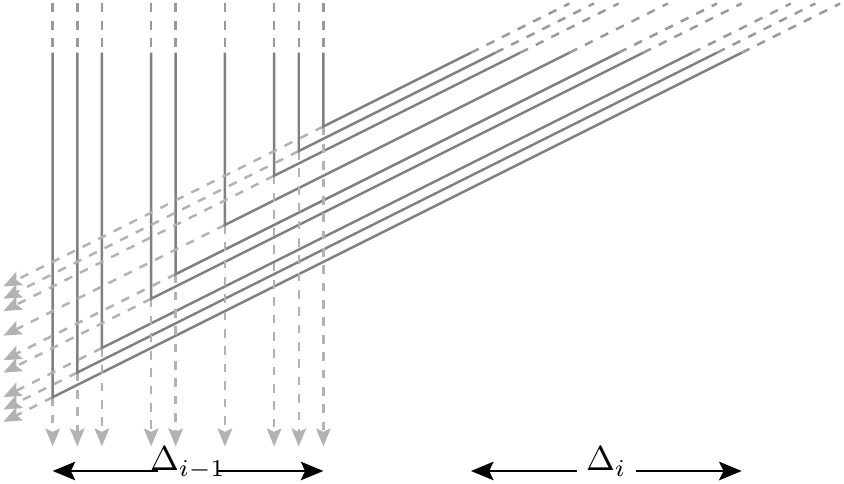}
  \caption{An example of how we join the lines for an interval $\Delta_{i-1}$
           with an interval $\Delta_i$. Note that the lines need not be drawn
           identically spaced in both intervals.}
  \label{fig:lem4}
\end{figure}

By the previous observation, this procedure therefore joins the lines leaving
$\Delta_{i-1}$ to those entering $\Delta_i$ and joins the lines leaving
$\Delta_i$ to those entering $\Delta_{i+1}$. Moreover, since the lines in each
odd-indexed $\Delta$-interval are left unrotated, we can repeat this procedure
for each even-indexed $\Delta$-interval. That is, for each $\Gamma_t$ with
$\mathbf{vertex}(\Gamma_t) = v$, the edges in $\mathbf{top}(\Gamma_t)$ are drawn
from $\gamma(v)$ to a set of vertical lines entering $\Delta_1$ in the same
left-to-right order as these edges were drawn incident to $v$ on the top page of
the book embedding. Similarly, the edges in $\mathbf{bottom}(\Gamma_t)$ are
drawn from $\gamma(v)$ to a set of lines leaving $\Delta_r$ in the same
left-to-right order as these edges were drawn incident to $v$ on the bottom
page. Furthermore, a line is drawn entering $\Delta_1$ and leaving $\Delta_r$
for each edge that had crossed the spine in the book embedding.

To complete the desired embedding, we consider the vertical lines entering
$\Delta_1$ and the lines leaving $\Delta_r$. The vertical lines entering
$\Delta_1$ correspond to the ends of the edges drawn on the top page of the book
embedding (either where they are incident to a vertex or where they cross the
spine). To join the two vertical lines corresponding to the same edge, we can
use the embedding of the top page of the book embedding. The procedure is as
follows. First, truncate the vertical lines at some common $y$-coordinate. Then,
draw the top page of the book embedding above the vertical lines, excluding the
region within some small distance $\epsilon$ from the spine (truncating the
edges before they meet their incident vertices). We can then trivially connect
the ends of the vertical lines to the ends of the truncated edges from the top
page since they occur from left to right in the same order. See
Fig~\ref{fig:lem5} for a depiction of this procedure.

\begin{figure}
  \centering
  \includegraphics[width=0.9\textwidth]{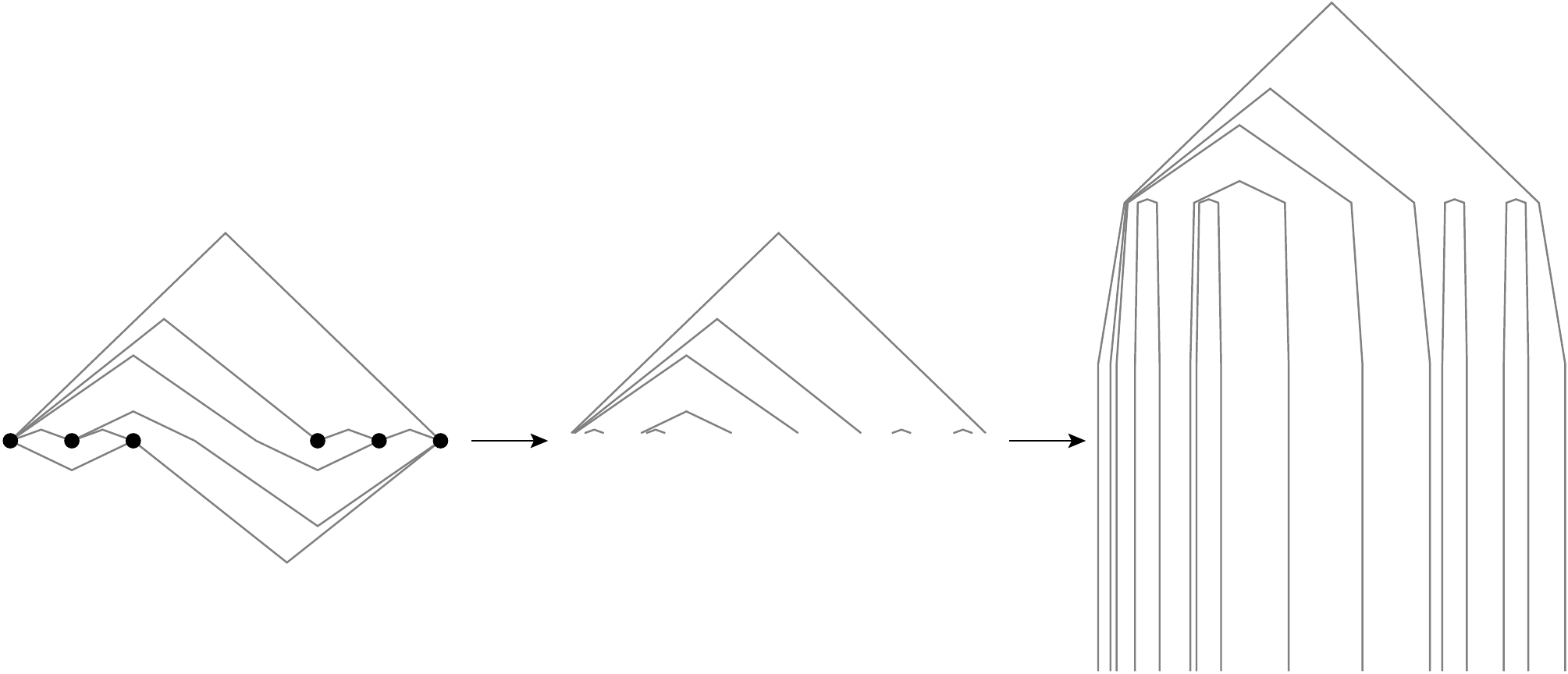}
  \caption{An example of how we join the vertical lines entering $\Delta_1$ by
           using the top page of the book embedding.}
  \label{fig:lem5}
\end{figure}

To join the edges leaving $\Delta_r$ with each other, we consider two cases. If
$r$ is odd, then the lines leave downwards, and we can connect them using the
bottom page in the same manner as we did with the lines entering $\Delta_1$. If
$r$ is even, then the lines leave upwards and we can again join them by using
the bottom page by simply rotating it to face the opposite direction.

We now bound the number of times an edge bends in the embedding. Recall that for
each of the $\Gamma_1,\Gamma_2,\dots,\Gamma_h$ we drew a set of piecewise linear
curves through the intervals $\Delta_1,\Delta_2,\dots,\Delta_r$. Each of these
curves bent at most once for each of the intervals and $O(1)$ times where they
connected to a vertex or joined another curve where they entered $\Delta_1$ or
left $\Delta_r$. By using the embedding from Observation~\ref{obs:book}, it
follows that each edge in $G$ is associated with at most three of
$\Gamma_1,\Gamma_2,\dots,\Gamma_h$. Thus, it follows that each edge bends at
most $3r + O(1)$ times. \qed
\end{proof}

\section{Drawing a Planar Graph with a Fixed Vertex Mapping}
\label{sec:applications}

In Section~\ref{sec:technique}, we established a technique for drawing 
a graph $G=(V,E)$ with a fixed vertex mapping $\gamma$, where the number of
bends is proportional to the size of a partition $V_1,V_2,\dots,V_r$ of $V$
satisfying the conditions of Lemma~\ref{lem:main}. In this section, we discuss
how to construct such a vertex partition for an arbitrary fixed vertex mapping.

Recall the definition of $\prec$ and $\succ$ from Definition~\ref{def:po}.
Clearly, any singleton set forms a chain with respect to both $\prec$ and
$\succ$. Thus, by Lemma~\ref{lem:main}, we can use the partition
$V_1=\{v_1\},V_2=\{v_2\},\dots,V_n=\{v_n\}$ to embed $G$ with any vertex mapping
using $3n + O(1)$ bends per edge. This bound's constant factor matches the best
known result of Badent et al. described in \cite{badent}. Using average-case
analysis, we can improve the bound.

A \emph{uniformly random planar graph} on $n$ vertices is a graph sampled
uniformly at random from the set of all planar graphs over the vertex set
$V=\{1,2,\dots,n\}$. We consider two isomorphic graphs to be different
if their vertex labelings differ. Suppose that we constructed a book
embedding of such a graph as per Observation~\ref{obs:book}. If we do so in a
manner that is independent from the labeling of the vertices, then we can assume
that the vertices occur along the spine of the book embedding in a uniformly
random order. To enforce independence, one could simply relabel the vertices
uniformly at random before constructing the book embedding, reverting to the
original labeling afterwards. Hence, we can make the following observation.

\begin{observation}
\label{obs:random_book}
A uniformly random planar graph $G$ can be embedded in the plane so that
\begin{enumerate}
  \item all vertices lie on a common line in a uniformly random order,
  \item each edge bends at most once above the line, at most once below the
        line, and at most once on the line.
\end{enumerate}
\end{observation}

By Observation~\ref{obs:random_book}, our analysis on random planar graphs
reduces to an analysis of random permutations. The proof of the next theorem
delineates this point.

\begin{theorem}
\label{thm:fixed}
A uniformly random planar graph $G=(V,E)$ can be embedded in the plane with a fixed
vertex mapping using at most $2n + O(1)$ bends per edge on expectation.
\end{theorem}

\begin{proof}
Let $\gamma$ be a fixed vertex mapping, and let $\delta$ be a direction for
which $\delta \cdot \gamma(u) = \delta \cdot \gamma(v)$, for $u,v \in V$, only
if $u = v$. Let $p_1, p_2, \dots, p_n$ be the points in the image of $V$ under
$\gamma$ in the order they occur along $\delta$. Define
$v_1=\gamma^{-1}(p_1),v_2=\gamma^{-1}(p_2),\dots,v_n=\gamma^{-1}(p_n)$.

Embed $G$ as per Observation~\ref{obs:random_book} so that the spine is aligned
with the direction $\delta$. For $i=1,\dots,n$, define $\alpha(v_i)$ to be the
index along $\delta$ at which $v_i$ occurs in this embedding. Thus, by
Definition~\ref{def:po}, $v_i \prec v_j$ if $i \leq j$ and
$\alpha(v_i) \leq \alpha(v_j)$. Similarly, $v_i \succ v_j$ if $i \leq j$ and
$\alpha(v_i) \geq \alpha(v_j)$.

By our choice of embedding, $\alpha(v_1),\alpha(v_2),\dots,\alpha(v_n)$ is a
uniformly random permutation of $1,2,\dots,n$. Construct a partition of $V$ as
follows. Let $t_1$ be the largest index such that
$\alpha(v_1),\alpha(v_2),\dots,\alpha(v_{t_1})$ is increasing. Then, let $t_2$
be the largest index such that
$\alpha(v_{t_1+1}),\alpha(v_{t_1+2}),\dots,\alpha(v_{t_2})$ is decreasing.
Repeat this process for $t=3,\dots,r$, maximizing increasing sequences when $i$
is odd and decreasing sequences when $i$ is even. The partition
$V_1=\{v_1,v_2,\dots,v_{t_1}\}$, $V_2=\{v_{t_1+1},v_{t_1+2},\dots,v_{t_2}\}$,
$\dots$, $V_r=\{v_{t_{r-1}+1},v_{t_{r-1}+2},\dots,v_{t_r}\}$ satisfies the
conditions of Lemma~\ref{lem:main} by construction. We can therefore construct
the desired embedding of $G$ so long as $r$ is at most $\frac{2}{3}n + O(1)$.

Thus, to complete the proof we consider how large $r$ is on average. Let $X$ be
the set of integers $1 < i < n$ for which
$\alpha(v_{i-1}),\alpha(v_{i+1}) < \alpha(v_i)$ or
$\alpha(v_{i-1}),\alpha(v_{i+1}) > \alpha(v_i)$. Clearly $r \leq |X| + 2$. Let
$X_2, X_3, \dots, X_{n-1}$ be indicator variables such that $X_i = 1$ if $i \in X$
and $X_i = 0$ otherwise. By linearity of expectation, it follows that
\[
\mathbf{E}|X| = \sum_{i=1}^n \mathbf{E}X_i
\]
and since $\mathbf{E}X_i = \mathbf{P}[X_i = 1] = 2/3$, it follows that
$r \leq 2/3(n+1)$. \qed
\end{proof}

\section{Simultaneous Embeddings with a Fixed Vertex Range}
\label{sec:se}

In this section, we consider the problem of embedding $k$ uniformly random
planar graphs $G_1,G_1,\dots,G_k$ over a common vertex set $V$, where the range
of the vertex mapping $\gamma$ is restricted to a fixed point set $P$ of size $n
= |V|$. As in the proof of Theorem~\ref{thm:fixed}, our analysis relies on
properties of uniformly random permutations.
\begin{lemma}[\cite{brightwell}]
\label{lem:bri}
Let $\pi_1,\pi_2,\dots,\pi_k$ be uniformly random permutations over the set
$S=\{1,2,\dots,n\}$. Then, there exists a partition $T_1,T_2,\dots,T_r$ of $S$, where
the elements in each part form an increasing subsequence in each of
$\pi_1,\pi_2,\dots,\pi_k$, such that $r$ is $O(n^{1 - \frac{1}{k+1}})$ with
overwhelming probability.
\end{lemma}
This bound was established by Brightwell in \cite{brightwell}. The following
result follows by combining this bound with Lemma~\ref{lem:main}.
\begin{theorem}
\label{thm:main}
If $G_1,G_2,\dots,G_k$ are uniformly random planar graphs, then we can embed
each graph in the plane with a common vertex mapping $\gamma : V \to P$ so that
all edges have $O(n^{1-\frac{1}{k}})$ bends each with overwhelming probability.
\end{theorem}

\begin{proof}
Let $\delta$ be a direction for which $\delta \cdot p = \delta \cdot q$, for
$p,q \in P$, only if $p = q$. Let $p_1, p_2, \dots, p_n$ be the points in $P$ in
the order they occur along $\delta$. Embed $G_1$ as per
Observation~\ref{obs:random_book} so that its spine is aligned with $\delta$.
Let $v_1,v_2,\dots,v_n$ be the vertices in $V$ in the order they occur along
$\delta$ in this embedding. Embedding $G_2,\dots,G_k$ in the same manner gives
the corresponding vertex orders $\pi_2, \pi_3, \dots, \pi_n$, where $\pi_i$ is a
uniformly random permutation of $v_1,v_2,\dots,v_n$. By Lemma~\ref{lem:bri}, it
follows that $V$ can be partitioned into $V_1,V_2,\dots,V_r$ such that the
vertices in $V_i$ occur along $\delta$ in the same order in the embeddings of
each of $G_1,G_2,\dots,G_k$. Furthermore, $r$ is $O(n^{1 - \frac{1}{k}})$ with
overwhelming probability. 

Let $\{u_1,\dots,u_{t_1}\} = V_1,\{u_{t_1+1},\dots,u_{t_2}\}=V_2,
\dots,\{u_{t_{r-1}+1},\dots,u_n\}=V_r$ such that $u_i$ occurs before $u_{i+1}$
along $\delta$ in the embedding of $G_1$, for all $i$, if $u_i,u_{i+1} \in V_j$
for some $j$. Consider the vertex mapping $\gamma$, defined such that
$\gamma(u_1)=p_1,\gamma(u_2)=p_2,\gamma(u_n)=p_n$.  By construction, each $V_i$
forms a chain with respect to $\prec$ from Definition~\ref{def:po}. Since
Lemma~\ref{lem:main} requires a partition that alternates between chains with
respect to $\prec$ and $\succ$, we can introduce empty sets into our partition
after each $V_i$, at most doubling its size. This extensions suffices since the
empty set forms a chain with respect to both $\prec$ and $\succ$. Moreover,
since $r$ is $O(n^{1 - \frac{1}{k}})$, so is this extension, and thus the claim
follows by Lemma~\ref{lem:main}. \qed
\end{proof}

\section{Lower Bounds on the Number of Bends}
\label{sec:lb}

In this section we prove that Theorem~\ref{thm:main} is optimal by using an
encoding argument. The proof relies on the following lemma.

\begin{lemma}
\label{lem:lb}
If the planar graph $G$ can be drawn in the plane with $\beta$ total bends under
a fixed vertex mapping that maps $V$ to a convex point set, then $G$ can be
encoded using \[ n \lg\left(\frac{\beta + n}{n}\right) + O(n)\]
bits.
\end{lemma}

The proof of this lemma is lengthy and is thus deferred to
Appendix~\ref{app:lb}. The following result follows by the information theoretic
lower bound on the number of bits required to encode a planar graph.

\begin{theorem}
Let $G_1,G_2,\dots,G_k$ be uniformly random planar graphs over the vertex set
$V$, and let $P$ be a convex point set of size $|V| = n$. Then, in all
simultaneous embeddings of $G_1,G_2,\dots,G_k$ that map $V$ to $P$, at least
one of $G_1,G_2,\dots,G_k$ has $\Omega(2^{1 - \frac{1}{k}})$ total bends with
overwhelming probability.
\end{theorem}

\begin{proof}
Suppose that $G_1,G_2,\dots,G_k$ can be drawn on $P$ with $\beta$ total bends
for some vertex mapping $\gamma$. Since there are $n!$ possible vertex mappings,
$\gamma$ can be encoded using $\lg{n!}$ bits. Thus, $G_1,G_2,\dots,G_k$ can be
encoded using
\[ kn\lg\left(\frac{\beta + n}{n}\right) + O(kn) + \lg{n!} \]
bits by Lemma~\ref{lem:lb}.  Since there are more than $n!$ planar graphs on $n$
vertices, it follows that at least $\lg{n!} - \Delta$ bits are required to encode
a uniformly random planar graph with probability at least $1 - 2^{-\Delta}$. It
follows that
\[ kn\lg\left(\frac{\beta + n}{n}\right) + O(kn) + \lg{n!} \geq k\lg{n!} - \Delta \]
with probability at least $1 - 2^{-\Delta}$. Thus, there exists a constant $c$
for which
\[ kn\lg\left(\frac{(\beta + n)2^{c + \frac{\Delta}{kn}}n^\frac{1}{k}}{n}\right) \geq kn\lg{n} \]
with probability at least $1 - 2^{-\Delta}$ by Stirling's approximation.
Dividing a factor of $kn$ and exponentiating both sides shows that the
inequality
\[ \frac{(\beta + n)2^{c + \frac{\Delta}{kn}}n^\frac{1}{k}}{n} \geq n \]
or equivalently,
\[ \beta \geq \frac{n^{2 - \frac{1}{n}}}{2^{c + \frac{\Delta}{kn}}} - n \]
holds with probability at least $1 - 2^{-\Delta}$.
\qed
\end{proof}

\bibliographystyle{splncs}
\bibliography{references}

\appendix

\section{Proof of Lemma~\ref{lem:lb}}
\label{app:lb}

In this section, we will demonstrate how to encode a planar graph $G$ with a
fixed vertex mapping $\gamma$ under the assumption that $G$ can be drawn on a
convex point set $P$ with $\beta$ bends. The encoding technique involves a
series of decompositions of $G$. First, we reduce the problem to encoding a
spanning tree of $G$. To encode this spanning tree, we will construct a
Hamiltonian 3-regular graph from the union of the tree's edges and the edges on
the convex hull of $P$.

The encoding of this Hamiltonian 3-regular graph follows by its recursive
structure; that is, we will describe a recursive method to encode it. This
recursive procedure makes use of edge separators. In particular, we will use an
edge separator that takes into account the crossing number of a graph.

\begin{lemma}[\cite{leighton}]
Let $G=(V,E)$ be a graph with nonnegative vertex weights that sum to at most 1
and do not individually exceed 2/3. Then, $G$ has an edge separator of size
\[  1.58 \sqrt{16\mathrm{cr}(G) + \sum_{v \in V} deg^2(v) } \]
where $\mathrm{cr}(G)$ is the crossing number of $G$.
\end{lemma}

Using such a separator, we can effectively decompose a Hamiltonian 3-regular
graph into 2 smaller Hamiltonian 3-regular graphs, separated by the edge separator.
The proof of the following lemma is essentially an analysis of this recursion.

\begin{lemma}
Let $G=(V,E)$ be a 3-regular graph with a fixed Hamiltonian cycle
$C=\{v_1v_2,v_2v_3,\dots,v_{n-1}v_n\}$. Then, $G$ can be encoded using
\[ \frac{n}{2}\lg\left(\frac{\mathrm{cr}(G) + n}{n}\right) + O(n) \]
bits, where $\mathrm{cr}(G)$ is the crossing number of $G$.
\end{lemma}

\begin{proof}
To prove the desired claim, we will prove the more precise bound
\[  \frac{n}{2}\lg\left(\frac{\sigma + n}{n}\right) + c_1n - c_2\sqrt{\sigma + n}\lg\left(\frac{n}{\sqrt{\sigma + n}}\right)  \]
by induction, where $\sigma$ is any parameter such that
$\sigma \geq \mathrm{cr(G)}$.

If $n \leq 2$, the claim holds trivially as no bits are required to encode $G$
and the bound we are trying to prove is nonnegative. Thus, we can proceed by
induction on $n$. 

Define $r = \sqrt{\sigma + n}$. We first make the assumption that $r \leq
\frac{n}{2^{45}}$. By the separator theorem, we can find a separator for $G$ of
size at most $1.58\sqrt{16\sigma + 9n}$ that partitions the vertices into two
sets $V_1$ and $V_2$, such that $\frac{1}{3}n \leq |V_1| \leq |V_2| \leq
\frac{2}{3}n$. Furthermore, we can assume that every edge in the separator was
in $C$ by at most doubling the size of the separator. Indeed, if any edge
$uv \in E \setminus C$  were such that $u \in V_1$ and $v \in V_2$, then we could add
$v$ to $V_1$, removing the edge $uv$ from the separator and adding at most two
edges incident to $v$ in $C$. Thus we can assume that the edge separator contains
only edges in $C$ and has size at most $2(1.58\sqrt{16\sigma + 9n}) \leq 13r$.

We can then encode $G$ recursively by encoding the two subgraphs induced by the
vertices in $V_1$ and the vertices in $V_2$ (after re-establishing the canonical
Hamiltonian cycle in each subgraph) and specifying how to combine these
subgraphs to construct $G$. It suffices to encode the edges in the separator, the
size of the separator, and a bit identifying which of $V_1$ or $V_2$ contained
$v_1$. This additional information costs at most
\[ \lg{ n \choose 13r} + \lg{13r} + 2 \]
bits, which by Stirling's approximation is at most $13r\lg\frac{n}{r}$ bits.
Thus, by induction, $G$ can be encoded using
\begin{equation}
\label{eqn:encoding_formula}
  T(n_1,\sigma_1) + T(n_2,\sigma_2) + 13r\lg\frac{n}{r}
\end{equation}
bits, where 
\[T(n,\sigma) = \frac{n}{2}\lg\left(\frac{\sigma + n}{n}\right) + c_1n - 
    c_2\sqrt{\sigma + n}\lg\left(\frac{n}{\sqrt{\sigma + n}}\right) \]
and $n_1,n_2,\sigma_1,\sigma_2$ satisfy
\[  n_1 + n_2 = n, \quad
    \sigma_1 + \sigma_2 = \sigma, \quad
    \frac{1}{3}n \leq n_1 \leq n_2 \leq \frac{2}{3}n. \]
Our goal thus is to bound the size of (\ref{eqn:encoding_formula}). Define $\alpha$ and $\lambda$ such that $\alpha n = n_1$ 
and $\lambda (\sigma + n) = \sigma_1 + n_1$. Thus, if follows that $(1-\alpha)n = n_2$ and $(1-\lambda)(\sigma + n) = \sigma_2 + n_2$. 
We can therefore express (\ref{eqn:encoding_formula}) as
\[  \frac{\alpha n}{2}\lg\left(\frac{\lambda(\sigma + n)}{\alpha n}\right) + c_1\alpha n - 
    c_2\sqrt{\lambda}r\lg\left(\frac{\alpha n}{\sqrt{\lambda}r}\right) \]
\[  {} + \frac{(1-\alpha) n}{2}\lg\left(\frac{(1-\lambda)(\sigma + n)}{(1-\alpha) n}\right) + c_1(1-\alpha) n - 
    c_2\sqrt{1 - \lambda}r\lg\left(\frac{(1 - \alpha) n}{\sqrt{1 - \lambda}r}\right) + 13r\lg\frac{n}{r} \]
bits or equivalently as
\[  \frac{\alpha n}{2}\lg\left(\frac{\lambda(\sigma + n)}{\alpha n}\right) + 
    \frac{(1-\alpha) n}{2}\lg\left(\frac{1-\lambda(\sigma + n)}{(1-\alpha) n}\right) + c_1n \]
\[  {} - c_2\left( \sqrt{\lambda}r\lg\left(\frac{\alpha n}{\sqrt{\lambda}r}\right) +
    \sqrt{1 - \lambda}r\lg\left(\frac{(1 - \alpha) n}{\sqrt{1 - \lambda}r}\right) - \frac{13}{c_2}r\lg\frac{n}{r} \right) \]
bits. To prove that this achieves the desired bound, we will consider 2 cases on the value of $\lambda$. First, assume that
$\frac{1}{6} \leq \lambda \leq \frac{5}{6}$. Observe that the function
\[  \frac{\alpha n}{2}\lg\left(\frac{\lambda(\sigma + n)}{\alpha n}\right) + 
    \frac{(1-\alpha) n}{2}\lg\left(\frac{(1-\lambda)(\sigma + n)}{(1-\alpha) n}\right) \]
achieves its maximum when $\lambda = \alpha$ and is therefore at most $\frac{n}{2}\lg\left(\frac{\sigma + n}{n}\right)$.
On the other hand, the function 
\[  \sqrt{\lambda}r\lg\left(\frac{\alpha n}{\sqrt{\lambda}r}\right) +
    \sqrt{1 - \lambda}r\lg\left(\frac{(1 - \alpha) n}{\sqrt{1 - \lambda}r}\right)  \]
achieves its minimum when $\lambda = \frac{1}{6}$ and $\alpha = \frac{1}{3}$ (without loss of generality). Thus, in this
case, the bound (\ref{eqn:encoding_formula}) is at most
\[  \frac{n}{2}\lg\left(\frac{\sigma + n}{n}\right) + c_1n \]
\[  {} - c_2\left( \sqrt{1/6}r\lg\left(\sqrt{2/3}\frac{n}{r}\right) +
    \sqrt{5/6}r\lg\left(\sqrt{8/15}\frac{n}{r}\right) - \frac{13}{c_2}r\lg\frac{n}{r} \right) \]
which simplifies to
\[  \frac{n}{2}\lg\left(\frac{\sigma + n}{n}\right) + c_1n - c_2r\lg{\frac{n}{r}} \]
\[  {} + c_2\left( \frac{13}{c_2}r\lg\frac{n}{r} + (\sqrt{1/6}\lg\sqrt{2/3} + \sqrt{5/6}\lg\sqrt{15/8})r - 
    (\sqrt{1/6} + \sqrt{5/6} - 1)r\lg\frac{n}{r} \right) \]
which is at most
\[  \frac{n}{2}\lg\left(\frac{\sigma + n}{n}\right) + c_1n - c_2r\lg{\frac{n}{r}} + 
    c_2\left( \frac{13}{c_2}r\lg\frac{n}{r} + \frac{3}{5}r - \frac{1}{3}r\lg\frac{n}{r} \right) \]
and, by setting $c_2 = 41$ is just
\[  \frac{n}{2}\lg\left(\frac{\sigma + n}{n}\right) + c_1n - c_2r\lg{\frac{n}{r}}  \]
since we assumed that $r \leq \frac{n}{2^{45}}$.

Next, we assume that $\lambda \leq \frac{1}{6}$. In this case, the function
\[  \frac{\alpha n}{2}\lg\left(\frac{\lambda(\sigma + n)}{\alpha n}\right) + 
    \frac{(1-\alpha) n}{2}\lg\left(\frac{(1-\lambda)(\sigma + n)}{(1-\alpha) n}\right) \]
achieves its maximum when $\lambda = \frac{1}{6}$ and $\alpha = \frac{1}{3}$ and is therefore at most
\[  \frac{n}{2}\lg\left(\frac{\sigma + n}{n}\right) - \frac{n}{18}. \]
Furthermore, the function
\[  \sqrt{\lambda}r\lg\left(\frac{\alpha n}{\sqrt{\lambda}r}\right) +
    \sqrt{1 - \lambda}r\lg\left(\frac{(1 - \alpha) n}{\sqrt{1 - \lambda}r}\right)  \]
is at least
\[  r\lg\frac{n}{r} - 2r  \]
since $\frac{1}{3} \leq \alpha \leq \frac{2}{3}$. Thus, it follows that the bound (\ref{eqn:encoding_formula}) is at most
\[  \frac{n}{2}\lg\left(\frac{\sigma + n}{n}\right) + c_1n - c_2r\lg{\frac{n}{r}} + 2r + 13r\lg\frac{n}{r} - \frac{n}{18} \]
which again is at most
\[  \frac{n}{2}\lg\left(\frac{\sigma + n}{n}\right) + c_1n - c_2r\lg{\frac{n}{r}} \]
as we have assumed that $r \leq \frac{n}{2^{45}}$. 

To complete the proof, we consider the case that $r > \frac{n}{2^{45}}$. That
is, we can assume that $\frac{n + \sigma}{n} > \frac{n}{2^{90}}$. As $G$ can be
encoded simply by encoding the edges in $E \setminus C$, it follows 
that $G$ can be encoded using $\frac{n}{2}\lg n$ bits or equivalently,
\[  \frac{n}{2}\lg{\frac{n}{2^{90}}} + 45n  \]
which in this case is at most
\[  \frac{n}{2}\lg\left(\frac{n + \sigma}{n}\right) + 45n  \]
bits. Furthermore, since 
\[  r\lg{\frac{n}{r}} \leq n \]
for all values of $r$, it follows that $G$ can be encoded using 
\[  \frac{n}{2}\lg\left(\frac{n + \sigma}{n}\right) + (45 + c_2)n - c_2n \leq \frac{n}{2}\lg\left(\frac{n + \sigma}{n}\right) + c_1n - c_2r\lg{\frac{n}{r}} \]
bits for any $c_1 \geq 86$ as we set $c_2 = 41$. 
\qed
\end{proof}

\begin{corollary}
\label{cor:tree}
Let $T$ be an ordered tree on $n$ vertices $V$, and let $P$ be a convex point
set. If $T$ can be drawn with a fixed vertex mapping $\gamma : V \to P$ such
that edges bend a total of $\beta$ times, then $T$ can be encoded with
\[ n\lg\left(\frac{\beta + n}{n}\right) + O(n) \]
bits.
\end{corollary}

\begin{proof}
Consider the boundary of the convex hull of the convex point set to which $V$ is
mapped by $\gamma$. In an embedding of $T$ with a total of $\beta$ bends under
this vertex mapping, the edges in $T$ can cross this boundary at most $2\beta$
times. Indeed, each piece composing an edge in $T$ can cross the boundary at
most twice, except those incident to a vertex, which can cross the boundary at
most once. Thus, if we construct a graph $G$ by the union of $T$ and the cycle
$C$ defined by the boundary of the convex hull of the point set, it follows that
$G$ has crossing number at most $2\beta$. Replace each vertex $v$ in $G$ with a
set of vertices, all of which are consecutive along $C$ and each of which is
incident to a unique edge that was incident to $v$ in $T$. This operation
produces a Hamiltonian 3-regular graph on $2n - 2$ vertices. Furthermore, as
this operation cannot increase the crossing number, it follows that the
resulting graph can be encoded with
\[ n\lg\left(\frac{\beta + n}{n}\right) + O(n) \]
bits. To recover $G$ from the encoding of this graph, it suffices to encode the
blocks of consecutive vertices along $C$ that corresponded to original vertices.
This can trivially be done using $O(n)$ bits. We can then recover the tree $T$
from $G$ as the vertex mapping is fixed. \qed
\end{proof}

To conclude this section, we describe how to use Corollary~\ref{cor:tree} to
encode an arbitrary planar graph $G$ drawn with a fixed vertex mapping into a
convex point set $P$ with $\beta$ total bends. Observe that $G$ can be assumed
to be connected. If not, we could make $G$ connected by introducing edges with
at most $O(\beta + n)$ total bends. In which case, the embedding of $G$ has an
ordered spanning tree $T$ as a subgraph that can be drawn with a fixed vertex
mapping into a convex point set using at most $O(\beta + n)$ total bends. By
Corollary~\ref{cor:tree}, $T$ can be encoded using 
\[ n\lg\left(\frac{\beta + n}{n}\right) + O(n) \]
bits.

We claim that only $O(n)$ bits are required to recover $G$ from the encoding of
$T$. Using the technique from \cite{pach}, we can introduce a Hamiltonian cycle
in $G$ by adding edges along a walk of the spanning tree $T$, starting from an
arbitrary vertex. The resulting Hamiltonian graph can be recovered from an
encoding of the Hamiltonian cycle and the encoding of the two outerplanar
graphs defining the edges inside and outside the Hamiltonian cycle. The
outerplanar graphs can be encoded using $O(n)$ bends (by a bijection with well
formed parenthesizations). Moreover, by construction, the order of the vertices
in the Hamiltonian cycle is recoverable from $T$. To construct the Hamiltonian
graph required introducing at most $O(n)$ subdivision vertices in $G$. Thus, we
can encode which vertices these corresponded to using $O(n)$ bits and recover
the original graph by removing them and the edges in the Hamiltonian cycle.

\end{document}